\documentclass[11pt]{article}
\usepackage{mathrsfs}
\usepackage{color}
\usepackage{enumerate}
\usepackage{mathrsfs}
\usepackage{amssymb}
\usepackage{amsfonts}
\usepackage{amsmath}
\usepackage{multirow}
\usepackage{amsthm}
\usepackage{url}
\usepackage{tikz}

\newtheorem{theorem}{Theorem}[section]

\newtheorem{lemma}[theorem]{Lemma}
\newtheorem{example}[theorem]{Example}

\newtheorem{proposition}[theorem]{Proposition}
\newtheorem{problem}{Problem}
\newtheorem{conjecture}[problem]{Conjecture}

\newtheorem{remark}{Remark}[section]

\def\c{\mathcal{C}}
\def\P{P}

\begin{document}
\title{ Maximum Distance Separable Codes for $b$-Symbol Read Channels}
\author{Baokun Ding$^{\text{a}}$, Tao Zhang$^{\text{a}}$ and Gennian Ge$^{\text{b,c,}}$\thanks{Corresponding author (e-mail: gnge@zju.edu.cn). Research supported by the National Natural Science Foundation of China under Grant Nos. 11431003 and 61571310.}\\
\footnotesize $^{\text{a}}$ School of Mathematical Sciences, Zhejiang University, Hangzhou 310027, China\\
\footnotesize $^{\text{b}}$ School of Mathematical Sciences, Capital Normal University, Beijing 100048, China\\
\footnotesize $^{\text{c}}$ Beijing Center for Mathematics and Information Interdisciplinary Sciences, Beijing 100048, China.\\}
\date{}
\maketitle
\begin{abstract}
  Recently, Yaakobi et al.\ introduced codes for $b$-symbol read channels, where the read operation is performed as a consecutive sequence of $b>2$ symbols. In this paper, we establish a Singleton-type bound on $b$-symbol codes. Codes meeting the Singleton-type bound are called maximum distance separable (MDS) codes, and they are optimal in the sense they attain the maximal minimum $b$-distance. Based on projective geometry and constacyclic codes, we construct new families of linear MDS $b$-symbol codes over finite fields. And in some sense, we completely determine the existence of linear MDS $b$-symbol codes over finite fields for certain parameters.
\end{abstract}

\smallskip
\section{Introduction}
In the traditional information theory, noisy channels are analyzed generally by dividing the message into individual information units. However, with the development of storage technologies, one finds that symbols cannot always be written and read consistently in channels that output overlapping symbols.

Cassuto and Blaum \cite{CB} first proposed a new coding framework for symbol-pair read channels. For a complete comprehension of the fruitful results on this topic, please refer to \cite{CB,CL,CJKWY,CLL,DGZZZ,KZL15,Li2016,YBS} and the references therein.
Recently, Yaakobi et al.\ \cite{YBH16} generalized the coding framework for symbol-pair read channels to that for $b$-symbol read channels, where the read operation is performed as a consecutive sequence of $b>2$ symbols. They also generalized some of the known results for symbol-pair read channels to those for $b$-symbol read channels.

This paper continues the investigation of  codes  for $b$-symbol read channels. We establish a Singleton-type bound for $b$-symbol codes, and construct several families of linear MDS $b$-symbol codes over finite fields  as follows:
\begin{enumerate}
  \item[(1)] there exists an MDS $(n,7)_{q}$ $3$-symbol code for $q$ being a prime power and $7\le n\le q^3+q^2+q+1$;
  \item[(2)] there exists an MDS $(n,9)_{q}$ $4$-symbol code for $q\ge 3$ being a prime power and $9\le n\le q^4+q^3+q^2+q+1$;
  \item[(3)] there exists an MDS $(n,2b+1)_{q}$ $b$-symbol code for $q$ being a prime power, $q\ge b\ge 5$ and $2b+1\le n\le q^{b}-bq^{b-1}+\frac{b^2+3b}{2}$;
  \item[(4)] there exists an MDS $(n,2b)_{q}$ $b$-symbol code with  $n\ge 2b$ for $q\ge b-1$ being a prime power, $b\ge 3$ or $q=2,b=4$;
  \item[(5)] there exists an MDS $(n,10)_{q}$ $5$-symbol code for $q\ge 3$ being a prime power and $n\ge 10$;
  \item[(6)] there exists an MDS $(n,2b)_{q}$ $b$-symbol code for $q$ being a prime power, $b\ge 5$, $n\ge 2b$ and $b|n$;
  \item[(7)] there exists an MDS $(\frac{q^{b+1}-1}{q-1},2b+1)_{q}$ $b$-symbol code for $q$ being prime power and any $b\ge 4$.
\end{enumerate}

The family  (4) indicates that a linear MDS $b$-symbol codes over $\mathbb{F}_{q}$ with $b=3,d_{3}=6$ or $b=4,d_{4}=8$ exists for any length $n$, $n\ge 2b$.  We also claim that a linear MDS $(n,2b+1)_q$ $b$-symbol code over $\mathbb{F}_{q}$ exists only when $2b+1\le n\le \frac{q^{b+1}-1}{q-1}$ (Lemma \ref{3.3}). Thus, in some sense, some of the families above completely determine the existence of linear MDS $b$-symbol codes over finite fields for certain parameters. Besides, we show that a linear MDS $(n,d_{b})_{q}$ $b$-symbol code with  $d_{b}<n$ is also an MDS $(n,d_{b}+1)_{q}$ $(b+1)$-symbol code (Theorem \ref{thmb+1}). Therefore, we can derive new MDS $b$-symbol codes from each family above.

This paper is organized as follows. In Section~\ref{pre} we present basic notations about $b$-symbol codes and derive a Singleton-type bound for $b$-symbol codes. In Section~\ref{conpro}, we construct MDS $b$-symbol codes from projective geometry. And in Section~\ref{concyc}, we give a construction of MDS $b$-symbol codes from constacyclic codes. Section~\ref{conclu} concludes the paper.
\section{Preliminaries}\label{pre}
Let $\Sigma$ be the alphabet consisting of $q$ elements, each element of which is called a symbol. Let $b$ be an integer and $b\ge 1$. For a vector ${\bf x}=(x_{0},x_{1},\cdots,x_{n-1})$ in $\Sigma^{n}$, we define the b-symbol read vector of ${\bf x}$ as
$$\pi_{b}({\bf x})=((x_{0},\cdots,x_{b-1}),(x_{1},\cdots,x_{b}),\cdots,(x_{n-1},x_{0},\cdots,x_{b-2}))\in(\Sigma^{b})^{n}.$$

Throughout this paper, let $q$ be a prime power and $\mathbb{F}_{q}$ be the finite field containing $q$ elements. We will focus on vectors over $\mathbb{F}_{q}$, so $\Sigma=\mathbb{F}_{q}$.
For two vectors ${\bf x}$, ${\bf y}$ in $\mathbb{F}_{q}^{n}$, we have $$\pi_{b}({\bf x}+{\bf y})=\pi_{b}({\bf x})+\pi_{b}({\bf y}),$$ and the $b$-distance between ${\bf x}$ and ${\bf y}$ is defined as
$$D_{b}({\bf x},{\bf y}):=|\lbrace 0\le i\le n-1:(x_{i},\cdots,x_{i+b-1})\neq (y_{i},\cdots,y_{i+b-1})\rbrace|,$$
 where the subscripts are reduced modulo $n$. Accordingly, the $b$-weight of ${\bf x}\in \mathbb{F}_{q}^n$ is defined as
$$wt_{b}({\bf x}):=|\lbrace 0\le i\le n-1:(x_{i},\cdots,x_{i+b-1})\neq {\bf 0}\rbrace|,$$
where the subscripts are reduced modulo $n$ and ${\bf 0}$ denotes the all-zeros vector. The Hamming distance between two vectors ${\bf x}$ and {\bf y} is denoted by $d_{H}({\bf x},{\bf y})$. Similarly, the Hamming weight of a vector ${\bf x}$ is denoted by $wt_{H}$({\bf x}). We have the following connection between the $b$-distance and the $b$-weight.

\begin{proposition}\label{prop1}
  For all ${\bf x},{\bf y}\in \mathbb{F}_{q}^n$, $D_{b}({\bf x},{\bf y})=wt_{b}({\bf x}-{\bf y}).$
\end{proposition}
\begin{proof}
 Note that for  ${\bf x}$, ${\bf y}$ in $\mathbb{F}_{q}^{n}$, we have $D_{b}({\bf x},{\bf y})=d_{H}(\pi_{b}({\bf x}),\pi_{b}({\bf y}))=wt_{H}(\pi_{b}({\bf x})-\pi_{b}({\bf y}))=wt_{H}(\pi_{b}({\bf x}-{\bf y}))=wt_{b}({\bf x}-{\bf y})$.
\end{proof}

Meanwhile, the connection between the Hamming weight and the $b$-weight was proven in \cite{YBH16} for vectors over the alphabet $\lbrace 0,1\rbrace$. Since the proof also works for vectors over  $\mathbb{F}_{q}$, we present the following proposition directly.
\begin{proposition}\label{prop2}
  Let ${\bf x}\in \mathbb{F}_{q}^{n}$ be such that $0<wt_{H}({\bf x})\le n-(b-1)$. Then, $$wt_{H}({\bf x})+b-1\le wt_{b}({\bf x})\le b\cdot wt_{H}({\bf x}).$$
\end{proposition}

Consider the $b$-weight and $(b+1)$-weight of a nonzero vector in $\mathbb{F}_{q}^n$ and the following proposition holds.
\begin{proposition}\label{prop3}
  For any nonzero vector ${\bf x}=(x_{0},x_{1},\cdots,x_{n-1})$ in $\mathbb{F}_{q}^n$ and $wt_{b}({\bf x})<n$, we have $wt_{b+1}({\bf x})\ge wt_{b}({\bf x})+1$.
\end{proposition}
\begin{proof}
  It is obvious that $wt_{b+1}({\bf x})\ge wt_{b}({\bf x})$, since if $(x_{i},\cdots,x_{i+b-1})\ne {\bf 0}$ then $(x_{i},\cdots,x_{i+b-1},x_{i+b})\ne {\bf 0}$, where the subscripts are reduced modulo $n$, for all $0\le i\le n-1$. We also have $(x_{j},\cdots,x_{j+b-1})={\bf 0}$ and $x_{j+b}\ne 0$ for some $0\le j\le n-1$ since $wt_{b}({\bf x})<n$. It follows that $(x_{j},\cdots,x_{j+b-1},x_{j+b})\ne {\bf 0}$, and thus $wt_{b+1}({\bf x})\ge wt_{b}({\bf x})+1$.
\end{proof}
For example, the Hamming weight of the four vectors $v_{1}=1110000, v_{2}=1100001, v_{3}=1101000,v_{4}=1010100$ are all $3$ while their $3$-weights equal $5,5,6,7$ respectively and their $4$-weights equal $6,6,7,7$ respectively. An obvious observation is that when the non-zero elements of a vector become closer, the $b$-weight tends to be smaller. And for vectors with fixed Hamming weight, one has the smallest $b$-weight when all the non-zero elements are in cyclically consecutive positions.

A code $\mathcal{C}$ over $\mathbb{F}_{q}$ of length $n$ is a nonempty subset of $\mathbb{F}_{q}^{n}$ and the elements of $\mathcal{C}$ are called codewords. The minimum $b$-distance of $\mathcal{C}$ is defined as
$$d_{b}=\min\lbrace{D_{b}(\bf x},{\bf y})\mid{\bf x},{\bf y}\in \mathcal{C},{\bf x}\neq {\bf y}\rbrace,$$
and the size of $\mathcal{C}$ is the number of codewords it contains. In general, a code $\mathcal{C}$ over $\mathbb{F}_{q}$ of length $n$, size $M$ and minimum $b$-distance $d_{b}$ is called an $(n,M,d_{b})_{q}$ $b$-symbol code. Note that the case $b=1$ is just the traditional codes that are widely studied. And the case $b=2$ corresponds to symbol-pair codes. Besides, if $\mathcal{C}$ is a subspace of $\mathbb{F}_{q}^{n}$, then $\mathcal{C}$ is called a linear $b$-symbol code. In this paper we focus on linear $b$-symbol codes over $\mathbb{F}_{q}$.

\begin{theorem}(Singleton Bound):
Let $q\geq 2$ and $b\leq d_{b}\leq n$. If $\mathcal{C}$ is an $(n,M,d_{b})_{q}$ $b$-symbol code, then we have $M\leq q^{n-d_{b}+b}$.
\end{theorem}
\begin{proof}
Suppose that $\mathcal{C}$ is an $(n,M,d_{b})_{q}$ $b$-symbol code with $q\geq 2$ and $b\leq d_{b}\leq n$. Delete the last $d_{b}-b$ coordinates from all the codewords in $\c$. Note that any $d_{b}-b$ consecutive coordinates contribute at most $d_{b}-1$ to the $b$-distance, thus the resulting vectors of length $n-d_{b}+b$ are still distinct since $\c$ has $b$-distance $d_{b}$. The conclusion follows from the fact that the maximum number of distinct vectors of length $n-d_{b}+b$ over  $\mathbb{F}_{q}$ is $q^{n-d_{b}+b}$.
\end{proof}

An $(n,M,d_{b})_{q}$ $b$-symbol code $\c$ with $M=q^{n-d_{b}+b}$ is called a maximum distance separable (MDS) $(n,d_{b})_{q}$ $b$-symbol code.

\begin{theorem}\label{thmb+1}
A linear MDS $(n,d_{b})_{q}$ $b$-symbol code $\c$ with  $d_{b}<n$ is also an MDS $(n,d_{b}+1)_{q}$ $(b+1)$-symbol code.
\end{theorem}
\begin{proof}
From Propositions \ref{prop1} and \ref{prop3}, we always have $d_{b+1}\geq d_{b}+1$, and thus $|\c|=q^{n-d_{b}+b} \geq q^{n-d_{b+1}+b+1}$. Together with the Singleton bound, we obtain the conclusion.
\end{proof}

Now, we are ready to give a sufficient condition for the existence of MDS $b$-symbol codes.
\begin{theorem}\label{MainThm}
There exists a linear MDS $(n,d+2b-2)_{q}$ $b$-symbol code $\mathcal{C}$ if there exists a matrix with $d+b-2$ rows and $n\ge d+2b-2\ge 2b$ columns over $\mathbb{F}_{q}$, denoted by $H=[H_{0},H_{1},\cdots,H_{n-1}]$, where $H_{i}$ $(0\leq i\leq n-1)$ is the $i$-th column of $H$, satisfying:
\begin{itemize}
\item[1.] any $d-1$ columns of $H$ are linearly independent;
\item[2.] there exist $d$ linearly dependent columns;
\item[3.] any  $d+b-2$ cyclically consecutive columns are linearly independent, i.e., $H_{i},H_{i+1},\cdots,H_{i+d+b-3}$ are linearly independent for $0\leq i\leq n-1$, where the subscripts are reduced modulo $n$.
\end{itemize}
\end{theorem}
\begin{proof}
Let $\c$ be the linear code with parity check matrix $H$. Then the first two conditions indicate that $\mathcal{C}$ is an $[n,n-d-b+2,d]$ linear code with size $q^{n-d-b+2}$. For an arbitrary codeword $c=(c_{0},c_{1},\cdots,c_{n-1})\in\mathcal{C}$, if there exists $j$ such that $c_{j}=c_{j+1}=\cdots=c_{j+b-2}=0$ and $c_{j+b-1}\neq 0$, where the subscripts are reduced modulo $n$, then one can consider the vector $v=(c_{j+b-1},\cdots,c_{n-1},c_{0},\cdots,c_{j+b-2})$. Rewrite $v$ as $v=(v_{0},v_{1},\cdots,v_{t},0,\cdots,0)$ for some $t\leq n-b$, where $v_{0},v_{t}\neq 0$. We also have $t\ge d+b-2$, since any  $d+b-2$ cyclically consecutive columns are linearly independent. Moreover, there are at least $d$ nonzero elements in the set  $\lbrace v_{0},v_{1},\cdots,v_{t}\rbrace$. It is easy to see $wt_{b}(c)=wt_{b}(v)\geq d+2b-2$. If there does not exist $j$ such that $c_{j}=c_{j+1}=\cdots=c_{j+b-2}=0$ and $c_{j+b-1}\neq 0$, then it is easy to see that $wt_{b}(c)=n$. Hence $d_{b}\ge d+2b-2$.
\end{proof}

\section{MDS $b$-symbol codes from projective geometry}\label{conpro}
 Let $V(r+1,q)$ be a vector space of rank $r+1$ over $\mathbb{F}_{q}$. The projective $r$-space over $\mathbb{F}_{q}$, denoted by $PG(r,q)$, is the geometry whose points, lines, planes, $\cdots$, hyperplanes are the subspaces of $V(r+1,q)$ of rank $1,2,3,\cdots,r$, respectively. The dimension of a subspace of $PG(r,q)$ is one less than the rank of a subspace of $V(r+1,q)$. We refer to \cite{P09} for more information on projective geometry.

Label the point of $PG(r,q)$ as $\langle(a_{0},a_{1},\cdots,a_{r})\rangle$, the subspace spanned by a nonzero vector $(a_{0},a_{1},\\\cdots,a_{r})$, where $a_{i}\in \mathbb{F}_{q}$ for $0\leq i\leq r$. Since these coordinates are defined only up to multiplication by a nonzero scalar $\lambda\in \mathbb{F}_{q}$ (here $\langle(\lambda a_{0},\lambda a_{1},\cdots,\lambda a_{r})\rangle=\langle(a_{0},a_{1},\cdots,a_{r})\rangle$), we refer to $a_{0},a_{1},\cdots,a_{r}$ as homogeneous coordinates.  Thus, the number of points in $PG(r,q)$ is given by $\frac{q^{r+1}-1}{q-1}$.

\begin{lemma}\label{lem1}
There exist $q+1$ hyperplanes in $PG(r,q)$ covering all the points in $PG(r,q)$ and  intersecting in a projective $(r-2)$-space.
\end{lemma}
\begin{proof}
Fix a projective $(r-2)$-space $U$ in $PG(r,q)$. Choose an arbitrary point $\P_{0}$ in $PG(r,q)\setminus U$, then $U$ and $\P_{0}$ generate a hyperplane $V_{0}$. Next, choose $\P_{1}$ in $PG(r,q)\setminus V_{0}$, which together with $U$ forms another hyperplane $V_{1}$. Repeat the procedure until all the points are covered. We obtain $q+1$ hyperplanes $V_{0},\cdots,V_{q}$, which intersect in $U$.
\end{proof}

In Theorem \ref{MainThm}, if we fix $d=3$, choose $n$ points in $PG(b,q)$, $b\geq 2$, and regard them as column vectors of the matrix $H$, then we have the following lemma.

\begin{lemma}\label{Thm1}
There exists a linear MDS $(n,2b+1)_{q}$ $b$-symbol code $\c$ if there exists a set $\mathcal{S}$ of $n\geq 2b+1$ points of $PG(b,q)$ satisfying the following conditions:
\begin{itemize}
\item[1.] there exist $3$ points in $\mathcal{S}$ lying on a line;
\item[2.] if the $n$ points are ordered, say $\P_{0},\P_{1},\cdots,\P_{n-1}$, then any $b+1$ cyclically consecutive points, i.e., $\P_{i},\P_{i+1},\cdots,\P_{i+b}$, where the subscripts are reduced modulo $n$, do not lie in a projective $(b-1)$-space for $0\leq i\leq n-1$.
\end{itemize}
\end{lemma}
Note that the first condition in Lemma \ref{Thm1} can be easily satisfied, thus we focus on ordering points in $PG(b,q)$ such that any $b+1$ cyclically consecutive points do not lie in a projective $(b-1)$-space.

  Since a nonzero element in a codeword can contribute at most $b$ to the $b$-weight, a $b$-symbol code whose minimum $b$-distance equals $2b+1$ must have the minimum Hamming distance being equal to or greater than $3$. In other words, the parity-check matrix of any linear MDS $(n,2b+1)_{q}$ $b$-symbol code should be of size $(b+1)\times n$ and has no two linearly dependent columns. Thus linear MDS $(n,2b+1)_q$ $b$-symbol codes exist only when $n\le \frac{q^{b+1}-1}{q-1}$.
  \begin{lemma}\label{3.3}
    A linear MDS $(n,2b+1)_q$ $b$-symbol code over $\mathbb{F}_{q}$ exists only when $2b+1\le n\le \frac{q^{b+1}-1}{q-1}$.
  \end{lemma}

\subsection{$b=2$}\label{sb1}
A projective plane $PG(2,q)$ is an incidence system of points and lines such that
\begin{itemize}
  \item For any two distinct points, there  is exactly one line through both.
  \item Any two distinct lines meet in exactly one point.
  \item There exist four points such that no three are collinear.
\end{itemize}
From Lemma \ref{lem1} we know that all the points in $PG(2,q)$ lie on $q+1$ lines, all of which intersect in a point, just as shown in Figure \ref{fig1}.
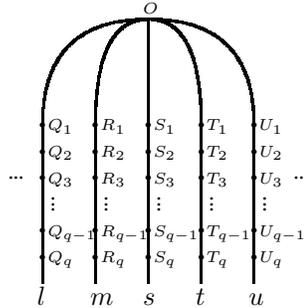
\begin{figure}[h]
\setlength{\unitlength}{1pt}
\centering
\begin{picture}(90,110)
\thicklines
{\tiny
\multiput(10,70)(20,0){5}{\line(0,-1){60}}
\put(50,110){\line(0,-1){40}}
\qbezier(50,110)(10,110)(10,70)
\qbezier(50,110)(30,110)(30,70)
\qbezier(50,110)(70,110)(70,70)
\qbezier(50,110)(90,110)(90,70)
\multiput(10,70)(20,0){5}{\circle*{1.5}}
\multiput(10,60)(20,0){5}{\circle*{1.5}}
\multiput(10,50)(20,0){5}{\circle*{1.5}}
\multiput(10,30)(20,0){5}{\circle*{1.5}}
\multiput(10,20)(20,0){5}{\circle*{1.5}}
\multiput(14,42)(20,0){5}{\circle*{1}}
\multiput(14,40)(20,0){5}{\circle*{1}}
\multiput(14,38)(20,0){5}{\circle*{1}}
\multiput(2,50)(-2,0){3}{\circle*{1}}
\multiput(106,50)(2,0){3}{\circle*{1}}

\put(48,112){$O$}
\put(12,68){$Q_{1}$}\put(32,68){$R_{1}$}\put(52,68){$S_{1}$}\put(72,68){$T_{1}$}\put(92,68){$U_{1}$}
\put(12,58){$Q_{2}$}\put(32,58){$R_{2}$}\put(52,58){$S_{2}$}\put(72,58){$T_{2}$}\put(92,58){$U_{2}$}
\put(12,48){$Q_{3}$}\put(32,48){$R_{3}$}\put(52,48){$S_{3}$}\put(72,48){$T_{3}$}\put(92,48){$U_{3}$}
\put(12,28){$Q_{q-1}$}\put(32,28){$R_{q-1}$}\put(52,28){$S_{q-1}$}\put(72,28){\tiny$T_{q-1}$}\put(92,28){$U_{q-1}$}
\put(12,18){$Q_{q}$}\put(32,18){$R_{q}$}\put(52,18){$S_{q}$}\put(72,18){$T_{q}$}\put(92,18){$U_{q}$}
\put(8,2){\small$l$}\put(28,2){\small$m$}\put(48,2){\small$s$}\put(68,2){\small$t$}\put(88,2){\small$u$}
}
\end{picture}
\caption {The structure of $PG(2,q)$\label{fig1}.}
\end{figure}

\begin{lemma}\label{thmb2}
  There exist $n$ ordered points in $PG(2,q)$ such that no three cyclically consecutive points are collinear for $q\ge 3$ being a prime power and $3\le n \le q^2+q+1$.
\end{lemma}
\begin{proof}
  Let the notations be as in Figure \ref{fig1}. There are many approaches to attain this goal and we give one of the strategies as follows.

    \noindent{\scriptsize $\bullet$} The case when $q$ is odd.

    In this case, through $O$ we have an even number of lines. Choose $O$ to be the first point, and then choose arbitrary
    points from lines $l$ and $m$ in turn. Suppose we have ordered the points as $O,Q_{1},R_{1},Q_{2},R_{2},\cdots,Q_{q},R_{q}$. Next we choose a point $S_{1}$ not on the line $Q_{q}R_{q}$ to be the next, and then a point $T_{1}$ not on the line $R_{q}S_{1}$. After that, choose points from lines $s$ and $t$ in turn. We can keep doing this until we have covered $n$ ($3\le n\le q^2+q+1$) points.

     Note that we have ordered $n$ points in $PG(2,q)$  and it is easy to check that no three consecutive points are collinear.  Denote the last three points as $\P_{n-3},\P_{n-2},\P_{n-1}$. We further need to make sure that $\P_{n-2}$, $\P_{n-1}$, $O$  are not collinear, neither are $\P_{n-1}$, $O$, $Q_{1}$. Since $\P_{n-1}$ is always not lying on the line $O\P_{n-2}$, we have $\P_{n-2}$, $\P_{n-1}$, $O$ are not collinear. Points $\P_{n-1}$, $O$, $Q_{1}$ may be collinear when $\P_{n-1}$ lies on the line $l$, i.e., $4\le n \le 2q$ and $n$ is even. If this happens we choose another point not lying on  lines $l,m$ and $\P_{n-3}\P_{n-2}$ to be the new last point, which can always succeed.

    \noindent{\scriptsize $\bullet$} The case when $q$ is even.

    This case is different from the case when $q$ is odd since there are an odd number of lines through $O$. For $n\le q^2+1$, we can choose an even number of lines and order the points on them just as what we do in the case when $q$ is odd. For $n> q^2+1$, we first put the points on the lines $l,m,s$ in order and then we can just proceed as in the case when $q$ is odd. Similarly we choose $O$ to be the first point, and then choose points from lines $l,m,s$ in turn, making sure that no three consecutive points are collinear. Since two lines meet in exactly one point, we can always do this until there is only one point left on each line. Suppose we have ordered the points as $O,Q_{1},R_{1},S_{1},Q_{2},\cdots,Q_{q-1},R_{q-1},S_{q-1}$. Choose $R_{q},S_{q}$ to be the next two points, after that, choose a point $T_{1}$ not on lines $R_{q}S_{q}$ and $Q_{q}S_{q}$ to be the next. Let the point $Q_{q}$ be the next, and then choose a point $U_{1}$ not on $Q_{q}T_{1}$, a point $T_{2}$ not on $Q_{q}U_{1}$. So far, we have ordered the points as $O,Q_{1},R_{1},S_{1},\cdots,Q_{q-1},R_{q-1},S_{q-1}R_{q},S_{q},T_{1},Q_{q},U_{1},T_{2}$ and no three consecutive points are collinear. There are an even number of lines left and we can then simply proceed as in the case when $q$ is odd.
\end{proof}
Note that in the lemma above, we exclude the case $q=2$. We show this in the following example.

\begin{example}\label{exam}
  We can order $3\le n\le 7$ points in $PG(2,2)$ such that any $3$ cyclically consecutive points are not collinear as shown in Table {\rm \ref{tab1}}, where the column vectors of $H$ denote the points.
  \begin{table}[h] \centering
  {\scriptsize
  \begin{tabular}{|c|c|}
    \hline
    $n$ & $H$ \\\hline
    $3$ & $\left(
             \begin{array}{ccc}
               0 & 1 & 1 \\
               1 & 1 & 0 \\
               0 & 0 & 1 \\
             \end{array}
           \right)$\\\hline
    $4$ & $\left(
             \begin{array}{cccc}
               0 & 1 & 1 &0\\
               1 & 1 & 0 &0\\
               0 & 0 & 1 &1\\
             \end{array}
           \right)$ \\\hline
    $5$ & $\left(
             \begin{array}{ccccc}
               0 & 1 & 1 &1&0\\
               1 & 1 & 0 &0&1\\
               0 & 0 & 1 &0&1\\
             \end{array}
           \right)$ \\\hline
    $6$ & $\left(
             \begin{array}{cccccc}
               0 & 1 & 1 &1&1&0\\
               1 & 1 & 0 &0&1&0\\
               0 & 0 & 1 &0&1&1\\
             \end{array}
           \right)$ \\\hline
    $7$ & $\left(
             \begin{array}{ccccccc}
               0 & 1 & 1 &1&0&0&1\\
               1 & 1 & 0 &0&1&0&1\\
               0 & 0 & 1 &0&1&1&1\\
             \end{array}
           \right)$ \\
    \hline
  \end{tabular}\caption{Ordered points in $PG(2,2)$.}\label{tab1}
  }\end{table}
\end{example}

%

\begin{remark}
 We have ordered $n$ points in $PG(2,q)$ for $q$ being a prime power and $3\le n\le q^2+q+1$. Actually, we can obtain MDS $(n,5)_{q}$ $2$-symbol codes for $q$ being a prime power with length $n$ ranging from $5$ to $q^2+q+1$ according to Lemma {\rm \ref{Thm1}}. We do not present this as a theorem since the result has been covered by our former paper {\rm \cite{DGZZZ}}. Besides providing a new proof of the result, the discussion in Lemma {\rm \ref{thmb2}} will also help in the proof of Lemma {\rm \ref{thmb4}}.
\end{remark}
\subsection{$b=3,d_{3}=7$}\label{sb2}
First, we collect some axioms for projective $3$-space, in which the objects (points, lines and planes) and the incidence relations are given:
\begin{itemize}
  \item Any two distinct points are incident with exactly one line.
  \item Any two distinct planes meet in exactly one line.
  \item Given any plane $\pi$ and any line $l$ not on $\pi$, there exists a unique point incident with both.
  \item Every plane incident with a given line $l$ is also incident with every point on $l$.
  \item Any two distinct lines meet in a point, if and only if they lie on a common plane.
  \item There exists a set of five points, of which no four lie on a common plane.
\end{itemize}

From Lemma \ref{lem1} we know that all the points in $PG(3,q)$ lie on $q+1$ planes, all of which intersect in a line, just as shown in Figure \ref{fig2}. For example, lines $l,l_{1},\cdots,l_{q}$ form a  plane, lines $l,m_{1},\cdots,m_{q}$ form another and the two planes share a common line $l$.
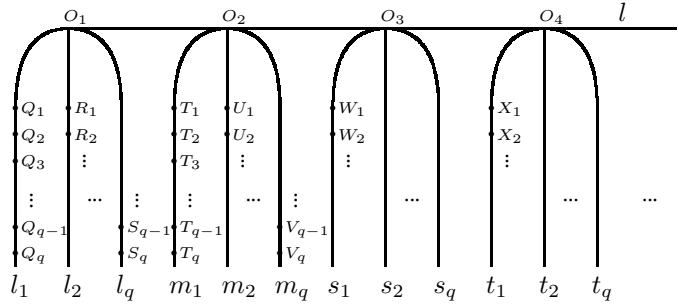
\begin{figure}[h]
\setlength{\unitlength}{1pt}
\centering
\begin{picture}(250,110)
\thicklines
{\tiny
\multiput(0,70)(20,0){12}{\line(0,-1){60}}
\multiput(20,100)(60,0){4}{\line(0,-1){30}}
\qbezier(20,100)(0,100)(0,70)
\qbezier(20,100)(40,100)(40,70)
\qbezier(80,100)(60,100)(60,70)
\qbezier(80,100)(100,100)(100,70)
\qbezier(140,100)(120,100)(120,70)
\qbezier(140,100)(160,100)(160,70)
\qbezier(200,100)(180,100)(180,70)
\qbezier(200,100)(220,100)(220,70)
\put(20,100){\line(1,0){230}}
\multiput(0,70)(60,0){4}{\circle*{1.5}}
\multiput(0,60)(60,0){4}{\circle*{1.5}}
\multiput(0,50)(60,0){2}{\circle*{1.5}}
\multiput(0,15)(60,0){2}{\circle*{1.5}}
\multiput(0,25)(60,0){2}{\circle*{1.5}}
\multiput(40,15)(60,0){2}{\circle*{1.5}}
\multiput(40,25)(60,0){2}{\circle*{1.5}}
\multiput(20,70)(60,0){2}{\circle*{1.5}}
\multiput(20,60)(60,0){2}{\circle*{1.5}}
\put(18,103){$O_{1}$}\put(78,103){$O_{2}$}\put(138,103){$O_{3}$}\put(198,103){$O_{4}$}\put(228,103){\small$l$}
\put(2,68){$Q_{1}$}\put(62,68){$T_{1}$}\put(122,68){$W_{1}$}\put(182,68){$X_{1}$}
\put(2,58){$Q_{2}$}\put(62,58){$T_{2}$}\put(122,58){$W_{2}$}\put(182,58){$X_{2}$}
\put(2,48){$Q_{3}$}\put(62,48){$T_{3}$}
\put(2,13){$Q_{q}$}\put(62,13){$T_{q}$}
\put(2,23){$Q_{q-1}$}\put(62,23){$T_{q-1}$}
\put(22,68){$R_{1}$}\put(82,68){$U_{1}$}
\put(22,58){$R_{2}$}\put(82,58){$U_{2}$}
\put(42,13){$S_{q}$}\put(102,13){$V_{q}$}
\put(42,23){$S_{q-1}$}\put(102,23){$V_{q-1}$}
\put(-2,0){\small $l_{1}$}\put(18,0){\small $l_{2}$}\put(38,0){\small $l_{q}$}
\put(58,0){\small $m_{1}$}\put(78,0){\small $m_{2}$}\put(98,0){\small $m_{q}$}
\put(118,0){\small $s_{1}$}\put(138,0){\small $s_{2}$}\put(158,0){\small $s_{q}$}
\put(178,0){\small $t_{1}$}\put(198,0){\small $t_{2}$}\put(218,0){\small $t_{q}$}
\multiput(6,37)(0,-2){3}{\circle*{1}}
\multiput(66,37)(0,-2){3}{\circle*{1}}
\multiput(28,35)(2,0){3}{\circle*{1}}
\multiput(88,35)(2,0){3}{\circle*{1}}
\multiput(148,35)(2,0){3}{\circle*{1}}
\multiput(208,35)(2,0){3}{\circle*{1}}
\multiput(238,35)(2,0){3}{\circle*{1}}
\multiput(26,52)(0,-2){3}{\circle*{1}}
\multiput(86,52)(0,-2){3}{\circle*{1}}
\multiput(126,52)(0,-2){3}{\circle*{1}}
\multiput(186,52)(0,-2){3}{\circle*{1}}
\multiput(46,37)(0,-2){3}{\circle*{1}}
\multiput(106,37)(0,-2){3}{\circle*{1}}
}
\end{picture}
\caption {The structure of $PG(3,q)$\label{fig2}.}
\end{figure}
\begin{lemma}\label{thmb3}
  There exist $n$ ordered points in $PG(3,q)$ such that no four cyclically consecutive points lie on a plane for $q\ge 3$ being a prime power and $4\le n \le q^{3}+q^2+q+1$.
\end{lemma}
\begin{proof}
 Fix a line $l$ and denote the $q+1$ planes intersecting in $l$ as $\pi_{0},\cdots,\pi_{q}$. In Figure \ref{fig2} we present four of them, and denote the plane corresponding to lines $l,l_{1},l_{2},\cdots,l_{q}$ as $\pi_{0}$, and the next three planes as $\pi_{1},\pi_{2},\pi_{3}$. We give one of the strategies as follows.
%

    \noindent{\scriptsize $\bullet$} The case when $q$ is odd.

    In this case, we have an even number of planes $\pi_{0},\cdots,\pi_{q}$ sharing the line $l$. Choose $O_{1}$, $O_{2}$ to be the first two points and then choose arbitrary points from lines $l_{1}$,$m_{1}$ in turn. Suppose we have ordered the points as $O_{1},O_{2},Q_{1},T_{1},Q_{2},\cdots,T_{q-1},Q_{q},T_{q}$. It is obvious that $O_{1},O_{2},Q_{1},T_{1}$ do not lie on a plane, neither do $O_{2},Q_{1},T_{1},Q_{2}$. For any other four consecutive points, we must have two of the points lying on $l_{1}$ and two on $m_{1}$. If they lie on a plane then the lines $l_{1}$ and $m_{1}$ must intersect. Suppose they meet in a point $O'$, then $O'$ lies on both $\pi_{0}$ and $\pi_{1}$, and thus on $l$, a contradiction. Therefore,  any four consecutive points do not lie on a plane.

    Next, we choose a point $R_{1}$ not on the plane $Q_{q}T_{q-1}T_{q}$, a point $U_{1}$ not on the plane $Q_{q}T_{q}R_{1}$, and a point $R_{2}$ not on the plane $T_{q}R_{1}U_{1}$. Then choose points from $l_{2},m_{2}$ in turn and we can proceed as above until all the points on lines $l_{1},\cdots,l_{q},m_{1},\cdots,m_{q}$ are covered.

    Suppose we have ordered the points as $O_{1},O_{2},Q_{1},T_{1},\cdots,S_{q-1},V_{q-1},S_{q},V_{q}$. Then we choose the following points to be $O_{3},O_{4},W_{1},X_{1},W_{2},X_{2}$, where $W_{1},W_{2}$ are arbitrary points on $s_{1}$ and  $X_{1},X_{2}$ are arbitrary points on $t_{1}$. It is easy to check that any four consecutive points are not on a plane. Repeat the procedure until we have covered $n$ ($4\le n\le q^3+q^2+q+1$) points in $PG(3,q)$.

    Note that we have ordered $n$ points in $PG(3,q)$ and no four consecutive points lie on a plane. Denote the last four points as $\P_{n-4},\P_{n-3},\P_{n-2}$ and $\P_{n-1}$, we further need to make sure:
    \begin{enumerate}
      \item[(1)] $\P_{n-1},O_{1},O_{2},Q_{1}$ do not lie on a plane.

      This fails only when $\P_{n-1}$ lies on $\pi_{0}$, i.e., $5\le n\le 2q^{2}+1$ and $n$ is odd. In this case, we choose a point not lying on planes $\pi_{0}$$,$ $\pi_{1}$$,$ $\P_{n-4}\P_{n-3}\P_{n-2}$ and $\P_{n-3}\P_{n-2}O_{1}$ to be the new last point $\P_{n-1}$.

      \item[(2)] $\P_{n-2},\P_{n-1},O_{1},O_{2}$ do not lie on a plane.

      This is always true in our construction, since $\P_{n-2},\P_{n-1}$ are always on different $\pi_{i}$s.
      \item[(3)] $\P_{n-3},\P_{n-2},\P_{n-1},O_{1}$ do not lie on a plane.

       If $\P_{n-3},\P_{n-1}$ lie on a line $l_{i}$, $1\le i\le q$, then we can fix this as what we do in case $(1)$. Otherwise, $\P_{n-3},\P_{n-2},\P_{n-1},O_{1}$ may lie on a plane only if they are chosen from different lines. For example, $T_{q},R_{1},U_{1}$ are chosen from lines $m_{1},l_{2},m_{2}$ respectively. In this case, we can always find a new suitable point $\P_{n-1}$ since there are enough points remaining.
    \end{enumerate}

        \noindent{\scriptsize $\bullet$} The case when $q$ is even.

    This case is different from the case when $q$ is odd since there are an odd number of planes. For $n\le q^{3}+q$, we can choose an even number of planes and proceed just as in the case when $q$ is odd. For $n>q^{3}+q$, we first order the points on the lines $l_{1},\cdots,l_{q},m_{1},\cdots,m_{q},s_{1},\cdots,s_{q}$ and then proceed  as in the case when $q$ is odd. Note that there are $3q$ lines, an even number, thus we can still consider the lines from different $\pi_{i}$s in pairs, $i=0,1,2$, and order the points as in the case when $q$ is odd. There are an even number of planes remaining. After a similar discussion, we can order $n$ points such that no four cyclically consecutive points are on a plane for $4\le n \le q^3+q^2+q+1$.
\end{proof}

\begin{example}\label{ex1}
  We can also order $4\le n\le 15$ points in $PG(3,2)$ such that any $4$ cyclically consecutive points do not lie on a plane as shown in Table {\rm \ref{tab2}}, where the first $n$ columns of $H$ denote the ordered $n$ points.
\begin{table}[h]  \centering
{\scriptsize  \begin{tabular}{|c|c|}
    \hline
    $n$ & $H$ \\\hline
    $5$ & $\left(
             \begin{array}{ccccc}
             1&1&0&1&1\\
             0&0&0&1&1\\
             1&0&1&0&0\\
             0&0&1&1&0\\
             \end{array}
           \right)$\\\hline
    $7$ & $\left(
             \begin{array}{ccccccc}
             0&0&0&1&0&1&1\\
             1&0&0&0&1&0&0\\
             0&1&0&0&0&1&0\\
             0&0&1&1&1&1&0\\
             \end{array}
           \right)$\\\hline
    $8$ & $\left(
             \begin{array}{cccccccc}
              1&1&1&0&1&0&0&0\\
              0&1&0&0&1&1&0&1\\
              1&0&1&1&0&1&0&0\\
              1&1&0&0&0&1&1&0\\
             \end{array}
           \right)$ \\\hline
    $10$ & $\left(
             \begin{array}{cccccccccc}
              1&0&1&1&0&0&1&1&1&0\\
              0&1&1&1&1&0&0&1&0&1\\
              0&1&0&0&0&1&1&1&1&1\\
              1&0&1&0&0&1&0&1&1&1\\
             \end{array}
           \right)$ \\\hline
    $13$ & $\left(
             \begin{array}{ccccccccccccc}
             1&1&0&1&0&0&1&0&1&1&0&0&0\\
             1&0&0&0&1&1&0&1&1&1&1&0&0\\
             1&1&0&0&0&1&1&0&0&1&1&1&1\\
             0&0&1&1&0&1&1&1&1&1&0&1&0\\
             \end{array}
           \right)$ \\\hline
    $4,6,9,11,12,14,15$ & $\left(
             \begin{array}{ccccccccccccccc}
             0&0&0&1&0&1&0&0&1&0&1&1&1&1&1\\
             1&0&0&0&1&0&1&1&0&0&1&1&0&1&1\\
             0&1&0&0&0&1&1&1&0&1&1&1&1&0&0\\
             0&0&1&1&1&1&0&1&0&1&0&1&0&1&0\\
             \end{array}
           \right)$ \\
    \hline
  \end{tabular}\caption{Ordered points in $PG(3,2)$.}\label{tab2}
  }\end{table}
\end{example}
Combining Lemmas \ref{Thm1}, \ref{thmb3} and Example \ref{ex1}, we have the following theorem.
\begin{theorem}
  There exists a linear MDS $(n,7)_{q}$ $3$-symbol code for $q$ being a prime power with length $n$ ranging from $7$ to $q^3+q^2+q+1$.
\end{theorem}

\subsection{$b=4,d_{4}=9$}\label{sb3}
From Lemma \ref{lem1} we know that all the points in $PG(4,q)$ lie in $(q+1)$ projective $3$-spaces, all of which intersect in a plane, just as shown in Figure \ref{fig3}. Lines $l_{0},l_{1},\cdots,l_{q}$ intersect in a point $O$ and form a plane $\pi$. Similarly, the sets of lines $\lbrace m_{11},m_{12},\cdots,m_{1q}\rbrace$, $\lbrace m_{21},m_{22},\cdots,m_{2q}\rbrace$,$\cdots$,$\lbrace m_{q1},m_{q2},\cdots,m_{qq}\rbrace$ form planes $\pi_{01},\pi_{02},\cdots,\pi_{0q}$ respectively. Planes $\pi,\pi_{01},\cdots,\pi_{0q}$ intersect in the line $l_{0}$ and they together form the projective $3$-space $V_{0}$. We totally have $q+1$ such projective $3$-spaces, denoted as $V_{0},V_{1},\cdots,V_{q}$, all of which form the projective space $PG(4,q)$ and intersect in the plane $\pi$.

Note that in $PG(2,q)$ we order points such that no three cyclically consecutive points are collinear. To attain this goal, we choose points from different lines by an interleaving technique. After that, in $PG(3,q)$, we order points such that no four cyclically consecutive points lie on a plane by choosing points from pairs of skew lines (lines do not intersect) alternatively and using the axiom that two lines on a projective plane must intersect.

Ordering points in $PG(4,q)$ such that no five cyclically consecutive points are in a projective $3$-space is more complicated. We only show the main idea in the proof of the following lemma.
\begin{figure}[h]
\setlength{\unitlength}{0.8pt}
\centering
\begin{picture}(250,300)
\thicklines
{\tiny
\multiput(100,105)(0,95){3}{\line(1,0){140}}
\put(5,200){\line(1,0){95}}
\qbezier(100,105)(5,105)(5,200)
\qbezier(100,295)(5,295)(5,200)

\multiput(100,75)(20,0){9}{\line(0,-1){45}}\multiput(120,105)(60,0){3}{\line(0,-1){30}}
\qbezier(120,105)(100,105)(100,75)\qbezier(120,105)(140,105)(140,75)
\qbezier(180,105)(160,105)(160,75)\qbezier(180,105)(200,105)(200,75)
\qbezier(240,105)(220,105)(220,75)\qbezier(240,105)(260,105)(260,75)

\multiput(100,170)(20,0){9}{\line(0,-1){45}}\multiput(120,200)(60,0){3}{\line(0,-1){30}}
\qbezier(120,200)(100,200)(100,170)\qbezier(120,200)(140,200)(140,170)
\qbezier(180,200)(160,200)(160,170)\qbezier(180,200)(200,200)(200,170)
\qbezier(240,200)(220,200)(220,170)\qbezier(240,200)(260,200)(260,170)

\multiput(100,265)(20,0){9}{\line(0,-1){45}}\multiput(120,295)(60,0){3}{\line(0,-1){30}}
\qbezier(120,295)(100,295)(100,265)\qbezier(120,295)(140,295)(140,265)
\qbezier(180,295)(160,295)(160,265)\qbezier(180,295)(200,295)(200,265)
\qbezier(240,295)(220,295)(220,265)\qbezier(240,295)(260,295)(260,265)
\multiput(128,50)(2,0){3}{\circle*{1}}
\multiput(188,50)(2,0){3}{\circle*{1}}
\multiput(248,50)(2,0){3}{\circle*{1}}\multiput(206,60)(4,0){3}{\circle*{1.5}}
\multiput(128,145)(2,0){3}{\circle*{1}}
\multiput(188,145)(2,0){3}{\circle*{1}}
\multiput(248,145)(2,0){3}{\circle*{1}}\multiput(206,155)(4,0){3}{\circle*{1.5}}
\multiput(128,240)(2,0){3}{\circle*{1}}
\multiput(188,240)(2,0){3}{\circle*{1}}
\multiput(248,240)(2,0){3}{\circle*{1}}\multiput(206,250)(4,0){3}{\circle*{1.5}}
\multiput(282,115)(0,-5){3}{\circle*{2}}
\put(85,108){\small $l_{q}$}\put(85,203){\small $l_{1}$}\put(85,298){\small $l_{0}$}
\put(278,60){\large $V_{q}$}\put(278,155){\large $V_{1}$}\put(278,250){\large $V_{0}$}
\put(98,25){ $s_{11}$}\put(118,25){ $s_{12}$}\put(138,25){$s_{1q}$}
\put(158,25){$s_{21}$}\put(178,25){ $s_{22}$}\put(198,25){ $s_{2q}$}
\put(218,25){$s_{q1}$}\put(238,25){ $s_{q2}$}\put(258,25){ $s_{qq}$}
\put(98,120){$n_{11}$}\put(118,120){ $n_{12}$}\put(138,120){ $n_{1q}$}
\put(158,120){$n_{21}$}\put(178,120){ $n_{22}$}\put(198,120){ $n_{2q}$}
\put(218,120){$n_{q1}$}\put(238,120){ $n_{q2}$}\put(258,120){ $n_{qq}$}
\put(98,215){$m_{11}$}\put(118,215){ $m_{12}$}\put(138,215){ $m_{1q}$}
\put(158,215){$m_{21}$}\put(178,215){ $m_{22}$}\put(198,215){ $m_{2q}$}
\put(218,215){$m_{q1}$}\put(238,215){ $m_{q2}$}\put(258,215){ $m_{qq}$}
\put(-4,197){$O$}
\put(117,109){$O_{q1}$}\put(177,109){$O_{q2}$}\put(237,109){$O_{qq}$}
\put(117,204){$O_{21}$}\put(177,204){$O_{22}$}\put(237,204){$O_{2q}$}
\put(117,299){$O_{11}$}\put(177,299){$O_{12}$}\put(237,299){$O_{1q}$}
}
\end{picture}
\caption {The structure of $PG(4,q)$\label{fig3}.}
\end{figure}
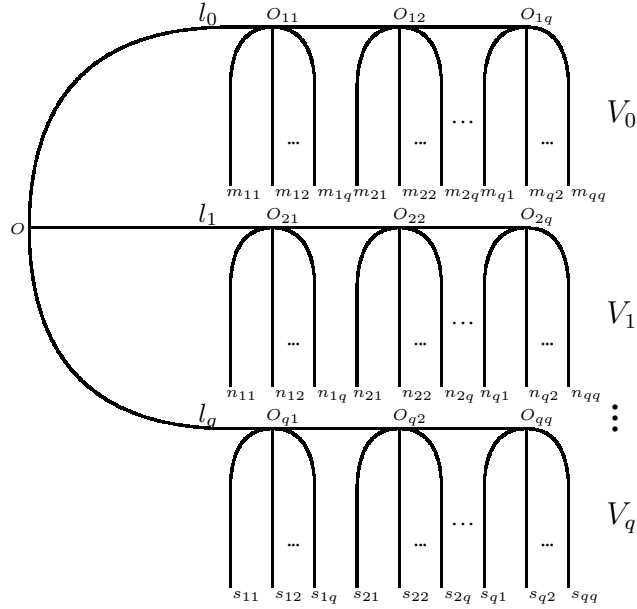

\begin{lemma}\label{thmb4}
    There exist $n$ ordered points in $PG(4,q)$ such that no five cyclically consecutive points are in a projective $3$-space for $q\ge 3$ being a prime power and $5\le n \le q^{4}+q^{3}+q^2+q+1$.
\end{lemma}
\begin{proof}
  Let the structure of $PG(4,q)$ be as shown in Figure \ref{fig3} and the notations be defined as above. Note that there are $q+1$ planes (including $\pi$) sharing a common line in each projective $3$-space $V_{i}$, $0\le i\le q$. Set aside the plane $\pi$ and denote the remaining $q$ planes in $V_{i}$ as $\pi_{i1},\cdots,\pi_{iq}$.  Let $O$ be the first point. For two lines on the same plane $\pi_{ij}$, for example $m_{11}$ and $m_{12}$, if we connect the point $O$ to every point of $m_{11}$ then we get $q+1$ lines, each of which intersects $m_{12}$ in exactly one point. Thus, we can build a one-to-one correspondence between the points on every two lines on the same plane $\pi_{ij}$. Consider the points lying on the lines in $\pi_{ij}\setminus l_{i}$, for example, points on $m_{11},\cdots,m_{1q}$ when $i=0,j=1$. We can easily order all the points such that no three consecutive points are collinear and no two consecutive points are collinear with $O$ for $q\ge 3$ after a similar discussion as in Lemma \ref{thmb2}.

  Choose two planes $\pi_{ij},\pi_{st},i\ne s$ and suppose we have ordered the points corresponding to the two planes as $P_{0},\cdots,P_{q^2}$ and $Q_{0},\cdots,Q_{q^2}$. We then order the points alternately as $P_{0},Q_{0},P_{1},Q_{1},\cdots,P_{q^2},\\Q_{q^2}$. Note that if we choose any five consecutive points, then three of them form the plane $\pi_{ij}$ or $\pi_{st}$ and the other two points form a line not through $O$. Thus they are not in a projective $3$-space, since in a projective $3$-space, every plane incident with a given line is also incident with every point on the line, and a line not on a plane must meet the plane in a point.

  The number of such planes $\pi_{ij}$ is $q(q+1)$, an even number. Thus we can always consider planes from different $V_{i}$s in pairs, $0\le i\le q$. Similar to Lemma \ref{thmb2} and Lemma \ref{thmb3}, we can repeat the procedure above until we have put $n$ ($5\le n \le q^{4}+q^{3}+q^2+q+1$) points in order and make sure that no five cyclically consecutive points lie in a projective $3$-space. We omit the tedious details here since the argument is analogous.
\end{proof}

Combining Lemma \ref{Thm1} and Lemma \ref{thmb4}, we have the following theorem.
\begin{theorem}
  There exists a linear MDS $(n,9)_{q}$ $4$-symbol code for $q\ge 3$ being a prime power with length $n$ ranging from $9$ to $q^4+q^3+q^2+q+1$.
\end{theorem}

\subsection{More constructions}\label{sb4}
 We first show the existence of MDS $b$-symbol codes for general $b\ge 5$ in the following theorem, which works quite well when $q$ is larger enough than $b$.
\begin{theorem}\label{bound1}
  There exists a linear MDS $(n,2b+1)_{q}$ $b$-symbol code for $q$ being a prime power, $q\ge b\ge 5$, with length $n$ ranging from $2b+1$ to $q^{b}-bq^{b-1}+\frac{b^2+3b}{2}$.
\end{theorem}
\begin{proof}
   From Lemma \ref{Thm1}, we mainly need to order $n$ points in $PG(b,q)$ such that any $b+1$ cyclically consecutive points do not lie in a projective $(b-1)$-space. We prove this theorem by induction. In $PG(b,q)$, we can easily find $b+1$ points that generate the whole space. Suppose we already have  $k$ ordered points, $b+1\le k< q^{b}-bq^{b-1}+2b$, denoted as $\P_{1},\P_{2},\cdots,\P_{k}$, such that any $b+1$ cyclically consecutive points do not lie in a projective $(b-1)$-space.

  Consider the $b+1$ projective ($b-1$)-spaces $V_{0},V_{1},\cdots,V_{b}$ generated by $\lbrace\P_{k-b+1},\P_{k-b+2},\cdots,\P_{k}\rbrace$, $\lbrace\P_{k-b+2},\P_{k-b+3},\cdots, \P_{k},\P_{1}\rbrace$, $\cdots$, $\lbrace\P_{1},\P_{2},\cdots,\P_{b}\rbrace$ respectively. We can always find a new suitable point $\P_{k+1}$ if the remaining points, i.e., points in $PG(b,q)\setminus\lbrace P_{1},P_{2},\cdots,P_{k}\rbrace$, are not all covered by the projective spaces $V_{0},V_{1},\cdots,V_{b}$.

  We determine the largest number of the remaining points covered by the $b+1$ spaces above. Two projective $(b-1)$-spaces in $PG(b,q)$ must intersect in a projective $(b-2)$-space. Thus $V_{0}$ covers $\frac{q^b-1}{q-1}$ points and any other $V_{i}$ covers at most $\frac{q^b-1}{q-1}-\frac{q^{b-1}-1}{q-1}$ new points for $1\le i\le q$. Besides, we should exclude the points $\P_{k-b+1},\P_{k-b+2},\cdots,\P_{k},\P_{1}, \cdots,\P_{b}$ when we count for each space. For example, we should exclude points $\P_{k-b+1},\P_{k-b+2},\cdots,\P_{k}$ when counting the points for $V_{0}$. And we  exclude only one point $P_{1}$ when counting for $V_{1}$ since points $\P_{k-b+2},\cdots,\P_{k}$ are in the intersection of $V_{0}$ and $V_{1}$. Note that some of the $2b$ points $\P_{k-b+1},\P_{k-b+2},\cdots,\P_{k},\P_{1},\cdots,\P_{b}$ may be the same when $k<2b$. And the total number of points we should exclude takes the minimum value $2b$ when $k=b+1$. Therefore, the $b+1$ projective spaces $V_{0},V_{1},\cdots,V_{b}$  can totally cover at most $(b+1)\frac{q^b-1}{q-1}-b\frac{q^{b-1}-1}{q-1}-2b$ points in $PG(b,q)\setminus\lbrace P_{1},P_{2},\cdots,P_{k}\rbrace$. And we can always find a new proper point $\P_{k+1}$ when $k< q^{b}-bq^{b-1}+2b$.

  Since $q\ge b$, we can always order $2b$ points such that any $b+1$ cyclically consecutive points do not lie in a projective $(b-1)$-space. Suppose we have ordered at least $2b$ points, i.e., $k\ge 2b$. In this case, no two of the $2b$ points $\P_{k-b+1},\P_{k-b+2},\cdots,\P_{k},\P_{1},\cdots,\P_{b}$ are the same and the largest number of points in $PG(b,q)\setminus\lbrace P_{1},P_{2},\cdots,P_{k}\rbrace$ covered by the $b+1$ projective $(b-1)$-spaces becomes $(b+1)\frac{q^b-1}{q-1}-b\frac{q^{b-1}-1}{q-1}-\frac{b^2+3b}{2}$. The conclusion follows.
  \end{proof}

 According to Theorem \ref{MainThm}, if we let $d=2$ and regard the columns of $H$ as vectors in $V(b,q)$, then we have the following lemma.

\begin{lemma}\label{thmd2}
There exists a linear MDS $(n,2b)_{q}$ $b$-symbol code $\c$ if there exists a set $\mathcal{S}$ of $n\geq 2b$ vectors of $V(b,q)$ satisfying:
\begin{itemize}
\item[1.] there exist $2$ linearly dependent vectors;
\item[2.] any $b$ cyclically consecutive vectors are linearly independent.
\end{itemize}
\end{lemma}
Similar to Theorem \ref{bound1}, we can derive the following theorem.
\begin{theorem}\label{thmn}
  There exists a linear MDS $(n,2b)_{q}$ $b$-symbol code with  $n\ge 2b$ for $q\ge b-1$ being a prime power, $b\ge 3$ or $q=2,b=4$.
\end{theorem}
\begin{proof}
  In a $b$-dimensional vector space $V(b,q)$, we can easily find $b$ vectors that generate the whole space. Suppose we already have $k\ge b$ ordered vectors, denoted as $v_{1},v_{2},\cdots,v_{k}$, such that any $b$ cyclically consecutive vectors are linearly independent.

  First, we consider the ($b-1$)-dimensional vector spaces $V_{1},V_{2},\cdots,V_{b}$ generated by $\lbrace v_{k-b+2},v_{k-b+3},\\\cdots,v_{k}\rbrace$, $\lbrace v_{k-b+3},v_{k-b+4},\cdots,v_{k},v_{1}\rbrace$, $\cdots$, $\lbrace v_{1},v_{2},\cdots,v_{b-1}\rbrace$ respectively. Two $(b-1)$-dimensional vector spaces in $V(b,q)$ must intersect in a $(b-2)$-dimensional vector space. Next we determine the largest number of nonzero vectors covered by the spaces above. $V_{1}$ covers $q^{b-1}-1$ nonzero vectors and any other $V_{i}$ covers at most $q^{b-1}-q^{b-2}$ new non-zero vectors for $2\le i\le b$. Besides, we should exclude the vectors $v_{k-b+2},\cdots,v_{k},\cdots,v_{1},v_{b-1}$ when we count for each vector space. Thus they can totally cover at most $bq^{b-1}-(b-1)q^{b-2}-2(b-1)-1$ nonzero vectors. We can always find a new proper vector $v_{k+1}$ unless all the non-zero vectors are covered by the $b$ vector spaces. In other words, we can always find a new proper vector if $q^b-bq^{b-1}+(b-1)q^{b-2}+2(b-1)\ge 1$, which turns out to be $q\ge b-1$ or $q=2,b=4$.
\end{proof}

 In the previous subsections we have given strategies to order $n$ vectors in the projective space $PG(b,q)$ such that any $b+1$ cyclically consecutive vectors are linearly independent for $b=2,3,4$ and $2b+1\le n\le \frac{q^{b+1}-1}{q-1}$. The following conclusion shows that $q\ge b-1$ is not an essential condition in Theorem \ref{thmn} if we order the vectors carefully.

 \begin{theorem}\label{thmb5}
There exists a linear MDS $(n,10)_{q}$ $5$-symbol code for $q\ge 3$ being a prime power and $n\ge 10$.
\end{theorem}
\begin{proof}
  For $n\ge 10$, we can find integers $n_{1},n_{2},\cdots,n_{t}$ such that $n=n_{1}+n_{2}+\cdots+n_{t}$, where $t\ge 2$ and $5\le n_{i}\le \frac{q^5-1}{q-1}$. From the conclusion in Subsection \ref{sb3}, we can find $t$ sequences of ordered points in $PG(4,q)$, denoted as $S_{1},S_{2},\cdots,S_{t}$, each of which has $n_{i}$ points and any $5$ cyclically consecutive points are linearly independent. Let the first $4$ points of the $t$ sequences be the same, which can be easily satisfied. Concatenating the $t$ sequences, we get a sequence of length $n$ satisfying the conditions in Lemma \ref{thmd2}.
\end{proof}

It is obvious that we can always find $b$ linearly independent vectors in a $b$-dimensional vector space $V(b,q)$. By the discussion in Theorem \ref{thmb5} we can conclude the following theorem.
\begin{theorem}
There exists a linear MDS $(n,2b)_{q}$ $b$-symbol code for $q$ being a prime power, $b\ge 5$, $n\ge 2b$ and $b|n$.
\end{theorem}

\section{MDS $b$-symbol codes from constacyclic codes}\label{concyc}
For $\eta\in\mathbb{F}_{q}^{*}$, a $q$-ary linear code $C$ of length $n$ is called $\eta$-constacyclic if it is invariant under the $\eta$-constacyclic shift of $\mathbb{F}_{q}^{n}$:
$$(c_{0},c_{1},\cdots,c_{n-1})\rightarrow(\eta c_{n-1},c_{0},\cdots,c_{n-2}).$$
If we identify each codeword $c=(c_{0},c_{1},\cdots,c_{n-1})$ with its polynomial representation $c(x)=c_{0}+c_{1}x+\cdots+c_{n-1}x^{n-1}$, then an $\eta$-constacyclic code $C$ of length $n$ over $\mathbb{F}_{q}$ is  identified with an ideal of the quotient ring $\mathbb{F}_{q}[x]/\langle x^{n}-\eta\rangle$, and $xc(x)$ corresponds to an $\eta$-constacyclic shift of $c(x)$. Moreover, $\mathbb{F}_{q}[x]/\langle x^{n}-\eta\rangle$ is a principal ideal ring, and $C$ is generated by a monic divisor $g(x)$ of $x^{n}-\eta$. In this case, $g(x)$ is called the generating polynomial of $C$ and we write $C=\langle g(x)\rangle$.

Let $\eta\in\mathbb{F}_{q}$ be a primitive $r$-th root of unity. Since $\textup{gcd}(n,q)=1$, there exists a primitive $(rn)$-th root of unity $\omega$ in some extension field of $\mathbb{F}_{q}$ such that $\omega^{n}=\eta$. It can be verified that
$$x^{n}-\eta=\prod_{i=0}^{n-1}(x-\omega^{1+ir}).$$

Let $\Omega=\{1+ir|0\leq i\leq n-1\}$. For each $j\in\Omega$, let $C_{j}$ be the $q$-cyclotomic coset modulo $rn$ containing $j$. Let $C$ be an $\eta$-constacyclic code of length $n$ over $\mathbb{F}_{q}$ with generating polynomial $g(x)$. Then the set $Z=\{j\in\Omega|g(\omega^{j})=0\}$ is called the defining set of $C$. We can see that the defining set of $C$ is a union of some $q$-cyclotomic cosets modulo $rn$ and $\textup{dim}(C)=n-|Z|$.

Similar to cyclic codes, there exists the following BCH bound for constacyclic codes.

\begin{theorem}$(${\rm\cite{KZL15}} The BCH bound for constacyclic codes$)$
Let $C$ be an $\eta$-constacyclic code of length $n$ over $\mathbb{F}_{q}$, where $\eta$ is a primitive $r$-th root of unity. Let $\omega$ be a primitive $(rn)$-th root of unity in an extension field of $\mathbb{F}_{q}$ such that $\omega^{n}=\eta$. Assume the generating polynomial of $C$ has roots that include the set $\{\omega^{1+ri}|i_{1}\leq i\leq i_{1}+d-2\}$. Then the minimum Hamming distance of $C$ is at least $d$.
\end{theorem}

Now we state our result.
\begin{theorem}
There exists a linear MDS $(\frac{q^{b+1}-1}{q-1},2b+1)_{q}$ $b$-symbol code for any $b\ge 4$ and $q$ being a prime power.
\end{theorem}
\begin{proof}
Let $n=\frac{q^{b+1}-1}{q-1}$, $\omega$ be a primitive element of $\mathbb{F}_{q}$ and $\delta$ be a primitive element of $\mathbb{F}_{q^{b+1}}$ such that $\delta^{n}=\omega$. Note that $g(x)=(x-\delta)(x-\delta^{q})\cdots(x-\delta^{q^{b}})\in\mathbb{F}_{q}[x]$ divides $x^{n}-\omega$. Let $C$ be the $\omega$-constacyclic code $\langle g(x)\rangle\subseteq\mathbb{F}_{q}[x]/(x^{n}-\omega)$. Then $C$ is an $[n,n-b-1,d]_{q}$ linear code with $3\leq d\leq b+2$.

If $d=b+2$, then it is easy to see that $d_{b}\geq2b+1$.

If $3\leq d\leq b+1$, let $c(x)=\sum_{i=0}^{n-1}c_{i}x^{i}$ be a nonzero codeword of $C$. If there exists $j$ such that $c_{j}=c_{j+1}=\cdots=c_{j+b-2}=0,c_{j+b-1}\neq0$, where the subscripts are reduced modulo $n$, then $x^{n-j-b+1}c(x)=\sum_{i=0}^{t}a_{i}x^{i}\in C$, for some $a_{i}\in \mathbb{F}_{q}$, $t\leq n-b$ and $a_{0},a_{t}\neq0$. Note that $3\leq d\leq b+1$ and $t\geq b+1$ since $g(x)|c(x)$, thus we have $wt_{b}(x^{n-j-b+1}c(x))=wt_{b}(c(x))\geq2b+1$. If there does not exist $j$ such that $c_{j}=c_{j+1}=\cdots=c_{j+b-2}=0,c_{j+b-1}\neq0$, then it is easy to see that $wt_{b}(c(x))=n$. Hence $d_{b}\geq2b+1$.
\end{proof}
\section{Conclusion}\label{conclu}
In this paper, we establish a Singleton-type bound for $b$-symbol codes and show that any linear MDS $b$-symbol code with $d_{b}<n$ is also an MDS $(b+1)$-symbol code. We give a sufficient condition for the existence of linear MDS $b$-symbol codes. And then, in specific cases, the problem turns out to be ordering points in $PG(b,q)$ such that no $b+1$ cyclically consecutive points lie in a projective $(b-1)$-space. As a result, we construct new families of linear MDS $b$-symbol codes with a large range of parameters and completely determine the existence of linear MDS $b$-symbol codes over finite fields for certain parameters.

This method is quite interesting and deserves further investigations. Consider the structure established by Lemma \ref{lem1}. Our goal is to order points in $PG(b,q)$ such that no $b+1$ cyclically consecutive points lie in a projective $(b-1)$-space. The main idea is as follows. For even $b$, in $PG(b,q)$, any two projective $\frac{b}{2}$-spaces in different projective $(b-1)$-spaces intersect in a point. For example, when $b=2$, any two of the $q+1$ lines meet in a point, and when $b=4$, $\pi_{ij}$ and $\pi_{st}$ meet in point $O$ ($i\ne s$). For a pair of projective $\frac{b}{2}$-spaces, we first order the points in each space separately such that any $\frac{b}{2}+1$ consecutive points generate the space (more details are omitted here), and then choose points alternatively from the pair of sequences of ordered points, just as what we do in Lemma \ref{thmb2} and Lemma \ref{thmb4}. For odd $b$, in $PG(b,q)$, any two projective $\frac{b-1}{2}$-spaces in different projective $(b-1)$-spaces have no points in common. For example, when $b=3$, in the structure established by Lemma \ref{lem1}, lines on different planes have no points in common. Similarly, for a pair of projective $\frac{b-1}{2}$-spaces, we first order the points in each space separately such that any $\frac{b-1}{2}+1$ consecutive points generate the space. Then we choose points alternatively from the pair of sequences of ordered points, just as what we do in Lemma \ref{thmb3}.

By the discussion above, it seems that we can give a strategy or an algorithm to order points in $PG(b,q)$ for any $b$ by induction, and thus can construct linear MDS $(n,2b+1)_{q}$ $b$-symbol codes for any $b$ and $2b+1\le n\le \frac{q^{b+1}-1}{q-1}$. However, we believe that such a proof will be tedious, and we prefer to present this as the following conjecture which calls for a neat and brief proof. We give more constructions in  Subsection \ref{sb4} and in Section \ref{concyc} to support the conjecture.
\begin{conjecture}
  There exist linear MDS $(n,2b+1)_{q}$ $b$-symbol codes for $q$ being a prime power, $b\ge 2$ and $2b+1\le n\le \frac{q^{b+1}-1}{q-1}.$
\end{conjecture}
Following the discussions and conclusions in Subsection \ref{sb4}, we also propose the following conjecture.
\begin{conjecture}
  There exist linear MDS $(n,2b)_{q}$ $b$-symbol codes for  $q$ being a prime power, $b\ge 2$ and $n\ge 2b$.
\end{conjecture}


\begin{thebibliography}{10}

\bibitem{CB}
Yuval Cassuto and Mario Blaum.
\newblock Codes for symbol-pair read channels.
\newblock {\em IEEE Trans. Inform. Theory}, 57(12):8011--8020, 2011.

\bibitem{CL}
Yuval Cassuto and Simon Litsyn.
\newblock Symbol-pair codes: Algebraic constructions and asymptotic bounds.
\newblock In {\em IEEE Int. Symp. Inf. Theory}, pages 2348--2352, 2011.

\bibitem{CJKWY}
Yeow~Meng Chee, Lijun Ji, Han~Mao Kiah, Chengmin Wang, and Jianxing Yin.
\newblock Maximum distance separable codes for symbol-pair read channels.
\newblock {\em IEEE Trans. Inform. Theory}, 59(11):7259--7267, 2013.

\bibitem{CLL}
B.~{Chen}, L.~{Lin}, and H.~{Liu}.
\newblock {Constacyclic symbol-pair codes: lower bounds and optimal
  constructions}.
\newblock {\em {\rm arXiv:1605.03460}}.

\bibitem{DGZZZ}
B.~Ding, G.~Ge, J.~Zhang, T.~Zhang, and Y.~Zhang.
\newblock New constructions of {MDS} symbol-pair codes.
\newblock {\em {\rm arXiv:1605.08859}}.

\bibitem{KZL15}
X.~Kai, S.~Zhu, and P.~Li.
\newblock A construction of new {MDS} symbol-pair codes.
\newblock {\em IEEE Trans. Inform. Theory}, 61(11):5828--5834, 2015.

\bibitem{Li2016}
S.~Li and G.~Ge.
\newblock Constructions of maximum distance separable symbol-pair codes using
  cyclic and constacyclic codes.
\newblock {\em Des. Codes Cryptogr.}, pages 1--14, 2016.
\newblock doi:10.1007/s10623-016-0271-y.

\bibitem{P09}
Stanley Payne.
\newblock Topics in finite geometry: ovals, ovoids and generalized quadrangles.
\newblock {\em UC Denver Course Notes}, 2009.

\bibitem{YBS}
E.~Yaakobi, J.~Bruck, and P.~H. Siegel.
\newblock Decoding of cyclic codes over symbol-pair read channels.
\newblock In {\em IEEE Int. Symp. Inf. Theory}, pages 2891--2895, 2012.

\bibitem{YBH16}
E.~Yaakobi, J.~Bruck, and P.~H. Siegel.
\newblock Constructions and decoding of cyclic codes over $b$-symbol read
  channels.
\newblock {\em IEEE Trans. Inform. Theory}, 62(4):1541--1551, 2016.

\end{thebibliography}
\end{document}